\documentclass[final,letterpaper]{IEEEtran}
\usepackage{amsmath,amsfonts}
\usepackage{color,graphicx,subfigure}

\newcommand{\tuplem}[1]{{#1}_1\;{#1}_2\;\cdots{#1}_m}
\newcommand{\tuplen}[1]{{#1}_1\;{#1}_2\;\cdots{#1}_n}

\newcommand{\SC}{\mathcal{C}}
\newcommand{\SB}{\mathcal{B}}
\newcommand{\T}[1]{\textrm{Tr}(#1)}
\newcommand{\tr}{\textrm{Tr}(\SC)}
\newcommand{\trp}{\textrm{Tr}(\SC^{\perp})}
\newcommand{\subs}{\textrm{SS}(\SC)}
\newcommand{\im}{\textrm{Im}_\SB (\SC)}
\newcommand{\I}[1]{\textrm{Im}_\SB (#1)}
\newcommand{\qm}{\textrm{GF}(q^m)}

\newcommand{\q}{\textrm{GF}(q)}
\newcommand{\two}{\textrm{GF}(2)}
\newcommand{\twom}{\textrm{GF}(2^m)}

\newtheorem{prop}{Proposition}
\newtheorem{example}{Example}

\begin{document}
\title{Reed-Solomon Subcodes with Nontrivial Traces: Distance Properties and Soft-Decision Decoding}
\author{Andrew~Thangaraj,~\IEEEmembership{Member,~IEEE,}
        and~Safitha~J~Raj%
\thanks{A. Thangaraj and S. J. Raj are with the Department of Electrical Engineering, Indian Institute of Technology Madras, Chennai 600036, India e-mail: andrew@iitm.ac.in.}}
\maketitle
\begin{abstract}
Reed-Solomon (RS) codes over GF$(2^m)$ have traditionally been the most popular non-binary codes in almost all practical applications. The distance properties of RS codes result in excellent performance under hard-decision bounded-distance decoding. However, efficient and implementable soft decoding for high-rate (about 0.9) RS codes over large fields (GF(256), say) continues to remain a subject of research with a promise of further coding gains. In this work, our objective is to propose and investigate $2^m$-ary codes with non-trivial binary trace codes as an alternative to RS codes. We derive bounds on the rate of a $2^m$-ary code with a non-trivial binary trace code. Then we construct certain subcodes of RS codes over GF($2^m$) that have a non-trivial binary trace with distances and rates meeting the derived bounds. The properties of these subcodes are studied and low-complexity hard-decision and soft-decision decoders are proposed. The decoders are analyzed, and their performance is compared with that of comparable RS codes. Our results suggest that these subcodes of RS codes could be viable alternatives for RS codes in applications.
\end{abstract}
\section{Introduction}
Reed-Solomon (RS) codes \cite{Reed:1960qf} are the most prevalent and commonly used codes today with applications ranging from satellite communications to computer drives. RS codes are popular, in theory, for their elegant algebraic construction. In practice, RS codes can be encoded and decoded with manageable complexity and high speed. RS codes continue to remain objects of active research with most recent interest being in list and soft-decision decoding \cite{Guruswami:1999df}\cite{Koetter:2003ph}\cite{Jiang:2008cr}.

Efficient soft decoding of RS codes has traditionally been a problem of importance. Early methods for soft decoding of RS codes included Chase decoding  and Generalized Minimum Distance (GMD) decoding \cite{Forney:1966qv}. Other methods for soft decoding RS codes include \cite{Vardy:1994fq}\cite{Ponnampalam:2002jy}. Recently, the Koetter-Vardy algorithm \cite{Koetter:2003ph}, the belief-propagation-based iterative algorithm \cite{Jiang:2004pz} and bit-level GMD algorithm \cite{Jiang:2008cr} have been proposed. Common themes in the above methods include (1) an additional coding gain of less than 1 dB, (2) an increase in complexity with size of the field, and (3) an increase in complexity for higher coding gain. As a result, efficient soft decoders are not readily available for high rate (rate 0.9 and above) RS codes over large fields (GF(256) and larger) in order to achieve the $\approx2-3$ dB of possible coding gain. The best coding gain achieved for the (255,239) RS code over GF(256) appears to be about 0.7 dB (at a block error rate of $10^{-3}$) over the additive white Gaussian noise (AWGN) channel. This is obtained with the bit-level GMD algorithm \cite{Jiang:2008cr}, which is a version of the GMD algorithm with bit-level erasures and Guruswami-Sudan list decoding.

In this work, we provide an approach for improving coding gain in the high-rate large-field case by exploiting the properties of images and traces of codes. Codes over GF$(2^m)$ are typically expanded into a binary image (using a basis for GF$(2^m)$ over GF(2)) before actual use in a physical channel. Hence, the binary image of codes over GF($2^m$) deserve to be studied closely. The binary trace code is closely associated with the image, since every image of a codeword over GF$(2^m)$ can be shown to be the concatenation of $m$ codewords from the trace code. However, the algebraic structure and distance properties of the image (and trace, to an extent) have proved to be difficult to characterize over the years. For instance, determining the basis that results in an image of highest minimum distance \cite{MacWilliams:1977rw} continues to be an open problem. Moreover, the exact practical utility (in terms of dBs of coding gain) of studying the properties of the trace and image have not yet been established concretely. One of the basic contributions of this work is to establish a possible utility,  in terms of coding gain, for studying the trace and image.

Specifically, in this work, we study codes over GF$(2^m)$ whose traces over GF(2) are non-trivial codes (not the identity code, for instance) with a minimum distance greater than 1. We characterize the structure of codes with a non-trivial trace and demonstrate properties that could be useful in practice. We derive some bounds on the minimum distance of the code and its trace using ideas of generalized Hamming weights \cite{Wei:1991if}. These bounds allow us to study the constraints on the minimum distance of the original code imposed by a non-trivial trace code.

On the practical side, we provide Reed-Solomon-like constructions for codes with a non-trivial trace. Basically, these are subcodes of RS codes whose traces are binary BCH codes. Suitable non-consecutive zeros are added to the set of zeros of a parent RS code to enable the trace to be a BCH code. We show that these codes, which we call Sub-Reed-Solomon (SRS) codes, meet the minimum distance bounds derived for codes with a non-trivial trace. Hence, SRS codes have best possible distance properties. In addition, our analysis (using list decoders) shows that a large fraction of errors beyond half the minimum distance are correctable. Hence, the performance of SRS codes is comparable to that of a traditional RS code at the same rate.

The main utility of SRS codes is that they are more amenable to efficient soft-decision decoding because of the trace structure. Since the image of a $2^m$-ary code is a concatenation of its binary trace, a soft decoder for the trace can be efficiently used to process soft input for the image. Using this idea, we propose simple soft decoders for SRS codes. Our simulations show that the proposed soft decoders for high-rate ($>0.9$) SRS codes over large fields (GF(256)) perform 0.4-0.5 dB better than other soft decoders of traditional RS codes at the same rate. A coding gain of 0.7-0.8 dB is possible over traditional bounded-distance decoders with low-complexity soft decoders, which involve efficient soft processing followed by traditional bounded-distance decoding. The complexity of obtaining 0.7 dB of coding gain with a (255,239) SRS code over GF(256) is lesser than that of the bit-level GMD algorithm running on the (255,239) RS code. Hence, our results suggest that SRS codes could be competent alternatives to RS codes in certain situations.

The rest of this article is organized as follows. Section \ref{sec:preliminaries} introduces the required notation and definitions for codes with a non-trivial trace. The basic structure of codes and their traces is shown in Section \ref{sec:structure}, and the bounds on minimum distance is discussed in Section \ref{sec:minim-dist-bounds-1}. Section \ref{sec:reed-solom-subc} introduces the construction of SRS codes and derives interesting properties of SRS codes. Hard-decision list decoders for SRS codes and their error-correcting properties are studied in Section \ref{sec:decoders-srs-codes}. Section \ref{sec:soft-input-decoders} discusses three different soft-input decoders for SRS codes. Finally, concluding remarks are made in Section \ref{sec:conclusion}.
\section{Notations and Basic Definitions}
\label{sec:preliminaries}
See \cite{MacWilliams:1977rw} for more details on the definitions and preliminary results in this section. A finite field $\qm$ ($q$: power of prime) is an $m$-dimensional vector space over $\q$. Trace of an element $\alpha\in \qm$ is a linear mapping $\textrm{Tr}:\qm\rightarrow \q$ defined by $\T{\alpha}=\sum_{i=0}^{m-1}\alpha^{q^i}$. The trace of a vector $[a_1\;a_2\cdots]$ is $[\T{a_1}\;\T{a_2}\cdots]$. If $\SC$ is a code over $\qm$, the trace of $\SC$, denoted $\tr$, consists of the traces of all codewords of $\SC$. In general, $\tr$ is a $(n,\geq k,\leq d)$ code over $\q$. The subfield subcode of $\SC$, denoted $\subs$, is defined as $\SC\bigcap(\q)^n$. $\subs$ contains the codewords of the $q^m$-ary code $\SC$ that are actually over $\q$. By Delsarte's theorem, we have 
\begin{equation}
  \label{eq:1}
  (\subs)^{\perp}=\trp.  
\end{equation}

A set of $m$ elements of $\qm$ linearly independent over $\q$ form a basis for this vector space. Let $\SB=\{\tuplem{\beta}\}$ be a basis for $\qm$ over $\q$. Let $\SB'=\{\tuplem{\beta'}\}$ be the dual basis of $\SB$ such that $\T{\alpha_i\beta_j}=\delta_{ij}$. Each element $\alpha\in \qm$ can be expanded as $\alpha=\sum^{m}_{i=1}a_i\beta_i$, where $a_i=\T{\alpha \beta'_i}$. The image of $\alpha\in \qm$ with respect to $\SB$ is the vector $\I{\alpha}=[\tuplem{a}]$ over $\q$. The image of $\SC$ with respect to $\SB$, denoted $\im$, consists of the images (with respect to $\SB$) of all codewords of $\mathcal{C}$. Image of an $(n,k,d)$ linear code over $\qm$ will be an $(nm,km,\geq d)$ linear code over $\q$.

For most cases in this paper, we restrict ourselves to $\twom$ for ease of description and practicality. Almost all our results have straight-forward extensions to $\qm$. Also, in the context of this paper, an $(n,k)$ code is said to be nontrivial if $1\le k\le n-1$.
\section{Structure}
\label{sec:structure}
Let $\SC$ be a linear code of length $n$ over $\qm$, and let $\SB=\{\tuplem{\beta}\}$ be a basis for $\qm$ over $\q$. Let $\SB'=\{\tuplem{\beta'}\}$ be the dual basis of $\SB$. 

Consider a codeword $\mathbf{c}=[c_{1}\;c_{2}\cdots c_{n}]\in\SC$. The image of $\mathbf{c}$ is the vector $[\I{c_1}\;\I{c_2}\cdots\I{c_n}]$, where $\I{c_i}=[\T{\beta'_1c_i}\;\T{\beta'_2c_i}\cdots\T{\beta'_mc_i}]$. For convenience, we view the image as a $n\times m$ matrix with the $i$-th row being $\I{c_i}$. 
\begin{prop}
Each column of an image matrix in $\im$ is a codeword of $\tr$.
\label{sec:preliminary-results-1}
\end{prop}
\begin{IEEEproof}
The $j$-th column of the image matrix will be $[\T{c_1\beta'_j}\;\T{c_2 \beta'_j}\cdots\T{c_n \beta'_j}]$.
$$\mathbf{c}\in\mathcal{C}\ \Rightarrow\ \beta'_j \mathbf{c}\in\mathcal{C}.$$
Hence the $j$-th column will belong to the trace of $\mathcal{C}$.
\end{IEEEproof}
The above property establishes the importance and utility of a non-trivial trace of a $q^m$-ary code. Basically, the image is a concatenation of codewords from the trace code with certain restrictions imposed by the overall code. As suggested by the concatenation, we let $\I{\mathbf{c}}=[\T{\beta'_1\mathbf{c}}\;\T{\beta'_2\mathbf{c}}\cdots\T{\beta'_m\mathbf{c}}]$, which is a permuted version of the image of $\mathbf{c}$. The image of $\SC$ is then defined as $\im=\{\I{\mathbf{c}}:\mathbf{c}\in\SC\}$.

The trace code imposes a structure on the party-check matrix of a $q^m$-ary code with a non-trivial trace.
\begin{prop}
Let $\SC$ be an $(n,k)$ code over $\qm$. Let $\tr$ be an $(n,k')$ code over $\q$ with a $n-k'\times n$ parity-check matrix $H'$. Then there exists a $n-k\times n$ parity-check matrix $H$ for $\SC$ of the form $$H=\begin{bmatrix}H'\\H''\end{bmatrix}.$$
\label{sec:preliminary-results-2}
\end{prop}
\begin{IEEEproof}
Since the rows of $H'$ belong to $\tr^{\perp}$, by Delsarte's theorem ~(\ref{eq:1}), the rows of $H'$ belong to $\textrm{SS}(\SC^{\perp})\subseteq \SC^{\perp}$. Since $H'$ is a full-rank matrix over $\q$ (and hence over $\qm$), the result follows.
\end{IEEEproof}
The matrix $H''$ will, in general, have entries from $\qm$. Starting from the parity-check matrix of Proposition \ref{sec:preliminary-results-2}, we can obtain a parity-check matrix for $\im$ with the form shown in Fig. \ref{fig:1}.
\begin{figure}
\begin{center}
\begin{tabular}{ccccc}
  $H'$&$\mathbf{0}$&$\mathbf{0}$&$\cdots$&$\mathbf{0}$\\
  $\mathbf{0}$&$H'$&$\mathbf{0}$&$\cdots$&$\mathbf{0}$\\
  $\vdots$&$\vdots$&$\vdots$&$\vdots$&$\vdots$\\
  $\mathbf{0}$&$\mathbf{0}$&$\mathbf{0}$&$\cdots$&$H'$\\
  \hline
  &&$H''_1$&&\\
  &&$H''_2$&&\\
  &&$\vdots$&&\\
  &&$H''_{k'-k}$&&
\end{tabular}
\caption{Structure of parity-check matrix for the image.}
\label{fig:1}
\end{center}
\end{figure}
In the matrix of Fig. \ref{fig:1}, 
$$H''_i=\begin{bmatrix}
\I{\beta_i\mathbf{h}''_1}\\
\I{\beta_i\mathbf{h}''_2}\\
\vdots\\
\I{\beta_i\mathbf{h}''_{k'-k}}
\end{bmatrix},\;1\leq i\leq m,
$$
where $\mathbf{h}''_j$, $1\leq j\leq k'-k$ denotes the $j$-th row of $H''$. It is clear that a nontrivial trace code imposes a useful structure on the parity-check matrix of the image. In this work, we exploit this structure for efficient soft decoding.
\section{Minimum Distance Bounds}
\label{sec:minim-dist-bounds-1}
We begin with a well-known basic result on the minimum distances of a code, its image and subfield subcode.
\begin{prop}
If $d$, $d_{ss}$ and $d_i$ are the minimum distances of $\SC$, $\subs$ and $\im$, respectively, we have $d\leq d_i\leq d_{ss}$.
\label{sec:preliminary-results-3}
\end{prop}
\begin{IEEEproof}
Clearly, $d_i\geq d$. Suppose $\mathbf{c}=[\tuplen{c}]\in\subs\subseteq\SC$ is a minimum weight codeword of $\subs$. Since $c_i\in \q$, image of $c_i\beta_1$ is 
$$[c_i\T{\beta_1\beta'_1}\ c_i\T{\beta_1\beta'_2}\ \cdots\ c_i\T{\beta_1\beta'_m}]=[c_i0\cdots0].$$
Hence, weight of the image of $\beta_1\mathbf{c}\in\SC$ is equal to the weight of $\mathbf{c}$, and the result follows.
\end{IEEEproof}
\subsection{Generalized Hamming weight bound}
The standard Singleton bound states that $d\leq n-k+1$ for a $(n,k,d)$ code $\SC$ over $\qm$. If we further require that the trace $\tr$ is a $(n,k',d')$ code, with $k'\geq k$, $d'\leq d$, the additional structure in the parity-check matrix results in a stronger bound on $d$. 

The notion of generalized Hamming weights (GHWs), introduced in \cite{Wei:1991if}, is used in the bound. Let $D$ be a subcode of a length-$n$ binary code $C$. The support of $D$, denoted $\chi(D)$, is defined as
$$\chi(D)=\{i:1\leq i\leq n,\exists[\tuplen{c}]\in D:c_i\ne0\}.$$
The set $\chi(D)$ is the set of positions where not all codewords in $D$ are zero. The $r$-th Hamming weight of $C$, denoted $d_r(C)$, is defined as 
$$d_r(C)=\min\{|\chi(D)|: D\text{ is a }(n,r)\text{ subcode of }C\}.$$  
In words, the $r$-th Hamming weight of $C$ is the minimum support of a $r$-dimensional subcode of $C$. 
\begin{prop}
Let $\SC$ be a $(n,k,d)$ code over $\qm$ with $\tr$ being a $(n,k')$ code ($k'\geq k$). Then,
$$d\leq d_{k'-k+1}(\tr)$$
\label{sec:minim-dist-bounds}
\end{prop}
\begin{proof}
Let
$$H=\begin{bmatrix}H'\\H''\end{bmatrix}$$
be a parity-check matrix for $\SC$ as per Proposition \ref{sec:preliminary-results-2}. Let $D$ be a $(k'-k+1)$-dimensional subcode of $\tr$ with support $\chi(D)$ such that $|\chi(D)|=d_{k'-k+1}(\tr)$. Let
$$H_D=\begin{bmatrix}H'_D\\H''_D\end{bmatrix}$$
be the submatrix of $H$ formed by the columns indexed by $\chi(D)$. The matrix $H'_D$, which is a parity-check matrix for $D$, has rank $r_D=|\chi(D)|-(k'-k+1)$. By row operations $H'_D$ can be reduced to the form 
$$\begin{bmatrix}
I_{r_D}&P'_D\\
\mathbf{0}&\mathbf{0}
\end{bmatrix},$$ 
where $I_{r_D}$ is the $r_D\times r_D$ identity matrix, $P_D$ is a $r_D\times (k'-k+1)$ matrix, and $\mathbf{0}$ represents all-zero matrices of suitable size. Therefore, by row operations, $H_D$ can be reduced to the form
$$\begin{bmatrix}
I_{r_D}&P'_D\\
\mathbf{0}&\mathbf{0}\\
\mathbf{0}&P''_D
\end{bmatrix},$$
where $P''_D$ is a $(k'-k)\times (k'-k+1)$ matrix with entries from $\qm$. Consider a $(k'-k+1)$-length vector $\mathbf{v}$ over $\qm$ such that $P''_D\mathbf{v}^T=\mathbf{0}$. From the form of $H_D$ above, it is clear that there exists a length-$r_D$ vector $\mathbf{u}$ such that $H_D[\mathbf{u}\;\mathbf{v}]^T=\mathbf{0}$. Hence, the vector with $[\mathbf{u}\;\mathbf{v}]$ in the positions $\chi(D)$ and zeroes for the remaining positions is a codeword of $\SC$ with weight less than or equal to $|\chi(D)|=d_{k'-k+1}(\tr)$. 
\end{proof}

For a $(n,k)$ code $C$ with a parity-check matrix $H$, another quantity closely related to generalized Hamming weights is the following, which is called equivocation with $s$ erasures ($0\leq s\leq n$) following \cite{Wei:1991if}:
\begin{equation}
  \label{eq:equi}
  \Delta_s(C)=\min_{I\subseteq\{1,2,\cdots,n\},|I|=s}\text{rank}(H_I),
\end{equation}
where $H_I$ denotes the submatrix of $H$ formed by the columns indexed by $I$. A careful reworking of Corollary A (Appendix) in \cite{Wei:1991if} shows that $\Delta_s(C)=\Delta$, $0\leq\Delta\leq n-k$ for which
\begin{equation}
  \label{eq:eqrange}
  n-d_{n-k-\Delta+1}(C^{\perp})<s\leq n-d_{n-k-\Delta}(C^{\perp})
\end{equation}
holds. Hence, the equivocations of a code $C$ can be computed using the generalized Hamming weights of the dual code $C^{\perp}$.
\begin{prop}
Suppose that $\SC$ is a $(n,k,d)$ code over $\qm$ with $\tr$ being a $(n,k',d')$ code over $\q$. Let $H=\left[\frac{H'}{H''}\right]$ be a parity-check matrix for $\SC$ such that $H'$ is a parity-check matrix for $\tr$. Then,
$$d\geq d''+\Delta_d(\tr),$$
where $d''$ is the minimum distance of the $(n,n-(k'-k))$ code over $\qm$ with parity-check matrix $H''$. 
\label{sec:gener-hamm-weight}
\end{prop}  
\begin{proof}
Suppose $\mathbf{c}\in\SC$ is a weight-$d$ codeword with nonzero positions $I\subseteq\{1,2,\cdots,n\}$, $|I|=d$. Let $H_I=\left[\frac{H'_I}{H''_I}\right]$ be the $n-k\times d$ submatrix of $H$ with columns indexed by $I$. 

By definition of $\Delta_d(\tr)$, we have that $\text{rank}(H'_I)\geq\Delta_d(\tr)$. By arguments similar to the proof of Proposition \ref{sec:minim-dist-bounds}, we see that row operations will result in $\Delta_d(\tr)$ columns of $H''$ becoming zero, and $d''\leq d-\Delta_d(\tr)$. 
\end{proof}

Since $H''$ has entries from $\qm$, we could meet the Singleton bound and have $d''=k'-k+1$ for several range of parameters. Assuming that the Singleton bound is met for the code with parity-check matrix $H''$, combining Propositions \ref{sec:minim-dist-bounds} and \ref{sec:gener-hamm-weight}, we get
$$k'-k+1+\Delta_d(\tr)\leq d\leq d_{k'-k+1}(\tr).$$
Hence, the generalized Hamming weights of $\tr$ and $\tr^{\perp}$ play a significant role in upper and lower bounding the minimum distance of a code with a non-trivial trace. 
\subsection{Sphere packing bound}
For the sphere packing bound, we restrict ourselves to the binary case and set $q=2$ for simplicity. As before, the image of a vector $\mathbf{v}=[v_{1}\;v_{2}\cdots v_{n}]\in\twom^n$ is represented as a $n\times m$ binary matrix $\I{\mathbf{v}}$ whose $i$-th row is $\I{v_i}=[v_{i1}\;v_{i2}\cdots v_{im}]$, $v_{ij}\in\two$. The $j$-th column of $\I{\mathbf{v}}$ is denoted $\bar{v}_j$.

Let $\SC$ be a $(n,k,d)$ code over $\twom$ with $t=\lfloor\frac{d-1}{2}\rfloor$. Let $d'$ be the minimum distance of $\tr$ with $t'=\lfloor\frac{d'-1}{2}\rfloor$. The sphere around a codeword $\mathbf{c}\in\SC$ is the following:
\begin{equation}
  \label{eq:sphere}
  S(\mathbf{c})=\{\mathbf{v}:d_H(\mathbf{v},\mathbf{c})\leq t\}\bigcup_{1\leq j\leq m}\{\mathbf{v}:d_H(\bar{v}_j,\bar{c}_j)\leq t'\}.
\end{equation}
As in standard sphere packing bounds, the sphere includes vectors in $\twom^n$ that are within a Hamming distance of $t$ from the codeword $\mathbf{c}$. In addition, vectors whose images contain columns that are within a Hamming distance of $t'$ from the corresponding column of the image of $\mathbf{c}$ are included in the sphere. 

Let $S_1=\{\mathbf{v}:d_H(\mathbf{v},\mathbf{c})\leq t\}$ and $S_2=\{\mathbf{v}:d_H(\bar{v}_j,\bar{c}_j)\leq t' (1\leq j\leq m)\}$. We see that $|S_1|=\sum^{t}_{l=0}\binom{n}{l}(2^m-1)^l$ and $|S_2|=\left(\sum^{t'}_{l=0}\binom{n}{l}\right)^m$. For $t\geq mt'$, $|S_1\cap S_2|=|S_2|$. 

For $t<mt'$, some additional combinatorics is involved in the computation of $|S_1\cap S_2|$. Let $\mathbf{v}_i$, $1\leq i\leq m$, be uniformly and independently chosen $n$-bit vectors of weight at most $t'$. Let the random variable $X_i=\text{wt}(\mathbf{v}_1\text{ OR }\mathbf{v}_2\text{ OR }\cdots\mathbf{v}_i)$, where $\text{wt}$ denotes Hamming weight and $\text{OR}$ denotes bitwise binary OR. The probability mass function (PMF) of $X_i$ can be iteratively found starting with $X_1$ and ending in $X_m$. Finally, 
$$|S_1\cap S_2|=|S_2|\text{Pr}\{X_m\leq t\}.$$ 

In either case, we let $|S_1\cap S_2|=\lambda_{nm}(t,t')|S_2|$, where $\lambda_{nm}(t,t')=1$ for $t\geq mt'$ and $\lambda_{nm}(t,t')=\text{Pr}\{X_m\leq t\}$ otherwise. Hence, $|S(\mathbf{c})|=(1-\lambda_{nm}(t,t'))|S_2|+|S_1|$, and the sphere-packing bound becomes
\begin{equation}
  \label{eq:spbound}
  2^{n-k}\geq(1-\lambda_{nm}(t,t'))\left(\sum^{t'}_{l=0}\binom{n}{l}\right)^m+\sum^{t}_{l=0}\binom{n}{l}(2^m-1)^l.
\end{equation}
Asymptotically, $\left(\sum^{t'}_{l=0}\binom{n}{l}\right)^m$ tends to $2^{nmH_2(t'/n)}$ and $\sum^{t}_{l=0}\binom{n}{l}(2^m-1)^l$ tends to $2^{nmH_{2^m}(t/n)}$, where $H_q(x)=x\log_q(q-1)-x\log_q x-(1-x)\log_q(1-x)$ is the entropy function. If $H_2(t'/n)>H_{2^m}(t/n)$, the first term dominates the bound, else the second term is dominant. Note that $\lambda_{nm}(t,t')=1$ when $t\geq mt'$ and the first term becomes zero.
\subsection{Existential bounds}
Bounds analogous to the Gilbert-Varshamov (GV) bound can be obtained for codes with a non-trivial trace. The traditional GV bound states that a $(n,k,d)$ code over $\twom$ exists whenever
$$(2^m)^{n-k}\geq\sum_{i=0}^{d-2}\binom{n-1}{i}(2^m-1)^i.$$
The RHS above is an upper bound on the number of $(n-k)$-tuples over $\twom$ that cannot be chosen as the $n$-th column of a parity-check matrix for a $(n,k,d)$ code. The $i$-th term in the RHS is the number of linear combinations of $i$ of the already-chosen $n-1$ columns. 

When the $(n,k,d)$ code has a $(n,k',d')$ trace, the form of the parity-check matrix results in a different upper bound on the tuples to be avoided in the $n$-th column. In this case, the parity-check matrix has the form $H=\left[\frac{H'}{H''}\right]$, where $H'$ is a $n-k'\times n$ binary matrix and $H''$ is a $k'-k\times n$ matrix over $\twom$. Let us suppose that $n-1$ columns of $H$ have been constructed and we attempt to add the $n$-th column. For $H'$, the constraint to maintain a distance $d'$ is the following:
\begin{equation}
  \label{eq:gvb}
2^{n-k'}\geq \sum_{i=0}^{d'-2}\binom{n-1}{i}.  
\end{equation}
Suppose the $n$-th column of $H'$, denoted $\bar{h'}_n$, has been chosen satisfying the above constraint. The number of $(k'-k)$-tuples over $\twom$ to be avoided in the $n$-th column of $H''$ can be bounded as follows. Consider a set $I\subseteq\{1,2,\cdots,n\}$, $|I|=i$. Let $H_I=\left[\frac{H'_I}{H''_I}\right]$ be the submatrix of $H$ with columns indexed by $I$. For $1\leq i\leq d'-2$, no linear combination of the columns of $H'_I$ can result in $\bar{h'}_n$ by \eqref{eq:gvb}. For $d'-1\leq i\leq d-2$, since the column rank of $H'_I$ is at least $d'-1$, a maximum of $(2^m-1)^{i-(d'-1)}$ linear combinations can result in $\bar{h'}_n$. Hence, an $n$-th column can be added for $H''$, whenever
\begin{equation}
  \label{eq:gvm}
  (2^m)^{k'-k}\geq\sum_{i=d'-1}^{d-2}\binom{n-1}{m}(2^m-1)^{i-(d'-1)}.
\end{equation}
Combining ~\eqref{eq:gvm} and ~\eqref{eq:gvb}, we get
\begin{multline}
  \label{eq:gv}
  n-k+(d'-1)\log_{2^m}(2^m-1)\geq \log_2\left(\sum_{i=0}^{d'-2}\binom{n-1}{i}\right)+\\
  \log_{2^m}\left(\sum_{i=d'-1}^{d-2}\binom{n-1}{m}(2^m-1)^i\right).
\end{multline}
An asymptotic version (large $m$, $n$) of the above bound, with $R=k/n$, is the following:
\begin{equation}
  \label{eq:gva}
  1-R+\frac{d'}{n}\geq H_{2^m}\left(\frac{d}{n}\right)+H_2\left(\frac{d'}{n}\right).
\end{equation}
\subsection{Illustration of bounds}
The GHW bound is difficult to compute in the general case, since strong bounds for generalized Hamming weights do not exist when the dimension grows with blocklength. In fact, the Singleton bound is seen to be tight in this case \cite{Helleseth:1995si}. 

For our purposes in this work, we compute the bounds discussed in this section for the case when (1) $n=255$, $d'=3$, $k'=247$ and (2) $n=255$, $d'=4$, $k'=246$ over GF(256). For this case, the corresponding trace codes are the (1) $(255,247,3)$ binary Hamming code and the (2) $(255,246,4)$ even-weight subcode of the Hamming code. For the Hamming code, generalized Hamming weights have been found exactly in \cite{Wei:1991if}. Let $C$ be the $(n=2^m-1,2^m-m-1,3)$ Hamming code. The generalized Hamming weights of $C$ are given by the following ordered set: 
\begin{multline}
  \label{eq:ghwh}
\{d_r(C):1\leq r\leq 2^m-m-1\}=\\
\{1,2,\cdots,n\}\setminus \{2^i:0\leq i< m\}.   
\end{multline}
If $C'$ is the even-weight subcode of $C$, the dual of $C'$ is the punctured Reed-Muller code. Hence,
\begin{multline*}
\{d_r(C'):1\leq r\leq 2^m-m-2\}=\\
\{2,\cdots,n\}\setminus \{1+2^i:0\leq i< m\}.
\end{multline*}

The bounds for the above two cases are shown in Fig. \ref{fig:bounds}. These bounds hold for any code over GF(256) with $n=255$ and a binary trace of minimum distance $d'=3$ and $d'=4$. The marks 'x' represent the standard Singleton bound $d\leq n-k+1$ without considering trace. The circle marks with legend 'SRS' represent points that are achieved by certain Subcodes of Reed-Solomon (SRS) codes that will be constructed in the later sections of this article. We see that the generalized Hamming weight bound is close to the standard Singleton bound. Hence, codes that achieve the generalized Hamming weight bound could be called 'trace-MDS'. The existential lower bound and the sphere-packing upper bound are shown as dotted lines in the figure.
\begin{figure}[htb]
  \centering
  \subfigure[$d'=3$]{
  \includegraphics[width=0.485\textwidth]{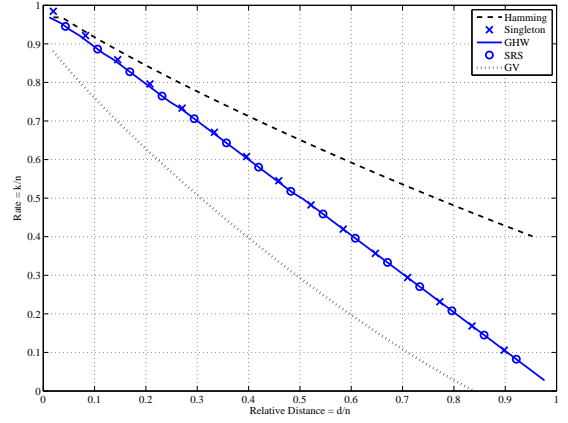}}
  \subfigure[$d'=4$]{
  \includegraphics[width=0.485\textwidth]{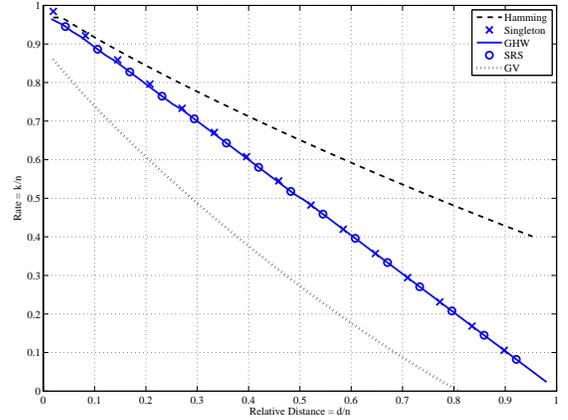}}
  \caption{Illustration of bounds for $n=255$, $m=8$.}
  \label{fig:bounds}
\end{figure}

From Fig. \ref{fig:bounds}, we see that the trace code does not significantly affect the rate when the minimum distance of the code ($d$) is reasonably larger than the minimum distance of the trace ($d'$). As can be expected, the upper and lower bounds are not very tight. This is because of the loose bounds on the combinatorial quantities in the derivation of the bound. These bounds could be improved in future work.

The points corresponding to SRS codes shown in Fig. \ref{fig:bounds} are seen to correspond to optimal codes over GF(256) with a trace code of minimum distance 3 and 4.

\section{Sub Reed-Solomon Codes}
\label{sec:reed-solom-subc}
In this section, we discuss the construction and basic properties of sub Reed-Solomon (SRS) codes with a nontrivial trace. We restrict ourselves to images of $\text{GF}(2^m)$ over GF(2) for simplicity. The construction easily extends to the general case. We will see that the SRS codes are trace-MDS in some cases, where the GHW bound can be evaluated. 
\subsection{Construction}
Let $\alpha$ be a primitive element of $\twom$. Let $\SC_{z}(t)$ denote the $(n,n-2t,2t+1)$ $t$-error correcting Reed-Solomon (RS) code of length $n=2^m-1$ with zero set be $Z_{rs}=\{z,z+1,\cdots,z+2t-1\}\mod n$. The generator polynomial of $\SC_z(t)$ is $\prod_{i=0}^{2t-1}(x+\alpha^{z+i})$. Typically, we let $z=0$ or $z=1$.

A SRS code $\SC_{zz'}(t,t')$ (for $t'\leq t$) is a subcode of $\SC_{z}(t)$ with zero set $Z_{rs}\cup Z_{bch}$, where $Z_{bch}$ is the zero set of a $t'$-error-correcting binary BCH code i.e.
$$Z_{bch}=C_{z'}\cup C_{z'+1}\cup\cdots\cup C_{z'+2t'-1},$$
where $C_i$ denotes the cyclotomic coset of $i$ modulo $n$ under multiplication by 2 and $z'\in Z_{rs}$. In the simplest examples, we choose $z=z'=1$ and denote $\SC_{11}(t,t')$ as simply $\SC(t,t')$. In some cases, we pick $z=0$ and $z'=1$.
\begin{example}
Let $\alpha$ be a primitive element of GF(256).
\begin{enumerate}
\item $\SC(8,1)$ is the subcode of the 8-error-correcting (255, 239, 17) RS code ($\SC(8)$) with zeros $\{1,2,\cdots,16,32,64,128\}$. $\SC(8,1)$ is a (255, 236, $\geq17$) code.
\item $\SC(8,2)$ is the subcode of the 8-error-correcting (255, 239, 17) RS code with zeros $\{1,2,\cdots,16,24,32,48,64,96,128,129,192\}$. $\SC(8,2)$ is a (255, 231, $\geq17$) code.
\item $\SC(6,1)$ is the subcode of the 6-error-correcting (255, 243, 13) RS code with zeros $\{1,2,\cdots,12,16,32,64,128\}$. $\SC(6,1)$ is a (255, 239, $\geq13$) code.
\item $\SC_{01}(6,1)$ is the subcode of the 6-error-correcting (255,243,13) RS code with zeros $\{0,1,2,\cdots,11,16,32,64,128\}$. $\SC_{01}(6,1)$ is a $(255,239,\geq13)$ code.
\end{enumerate}
\end{example}
\subsection{Properties}
The following properties can be proved for the SRS code $\SC_{zz'}(t,t')$ of length $n=2^m-1$ over GF$(2^m)$.
\begin{prop}
The trace of $\SC_{zz'}(t,t')$ is the binary cyclic code with zero set $Z_{bch}\cup Z'_{rs}$, where $Z'_{rs}\subseteq Z_{rs}$ is the largest possible union of cyclotomic cosets contained in $Z_{rs}$.
\label{sec:properties-1}
\end{prop}
\begin{IEEEproof}
This follows from Delsarte's theorem (\ref{eq:1}) and \cite[Chap 7 (Problem 33)]{MacWilliams:1977rw} 
\end{IEEEproof}
Thus, by Proposition \ref{sec:preliminary-results-1}, we see that when a codeword of the binary image of $\SC_{zz'}(t,t')$ is written down as a $n\times m$ matrix, each column will belong to the $t'$-error-correcting binary BCH code. When $z=1$, the trace will be equal to the BCH code in most practically relevant cases. However, when $z=0$, the trace will be the even-weight subcode of the $t'$-error-correcting BCH code. 

We now state a simple result about the subfield subcode of $C_{zz'}(t,t')$. This result is useful in finding the exact minimum distance of SRS codes.
\begin{prop}
The subfield subcode of the SRS code $\SC_{zz'}(t,t')$ is the binary cyclic code of length $n$ with zero set $\bigcup_{s\in Z_{rs}\cup Z_{bch}}C_{s}$. If $z=z'=1$, the subfield subcode is the $t$-error correcting BCH code with zeros $\bigcup_{s\in Z_{rs}}C_{s}$.
\label{sec:properties-2}
\end{prop}
\begin{IEEEproof}
This follows from \cite[Chap 7 (Problem 33)]{MacWilliams:1977rw}.
\end{IEEEproof}
As an example, consider the (255, 239, $\geq13$) code $\SC(6,1)$ over GF(256). The trace of the code is the length-255 binary Hamming code. The subfield subcode is the 6-error-correcting length-255 binary BCH code with exact minimum distance 13 \cite{Augot:1992cl}. Hence, $\SC(6,1)$ is a (255, 239, 13) code over GF(256). The (255,239,$\geq13$) code $\SC_{01}(6,1)$ over GF(256) has trace equal to the even-weight subcode of the length-255 binary Hamming code. Note that the minimum distance of the trace of $\SC_{01}(6,1)$ is 4.

Table \ref{tab:SRS} summarizes the parameters for some SRS codes that could have possible applications in practice. 
\begin{table*}[htb]
  \centering
  \begin{tabular}{|c|c|c|c|c|c|}
    \hline
    $\SC_{zz'}(t,t')$&$(n,k)$&$n-k+1$&$d$&$d'$&$Z_{bch}\cup Z_{rs}$\\
    \hline
    $\SC_{01}(6,1)$&(255,239)&17&13&4&$\{1,2,4,8,16,32,64,128\}\cup\{0,1,2,\cdots,11\}$\\
    $\SC_{01}(8,1)$&(255,235)&21&17&4&$\{1,2,4,8,16,32,64,128\}\cup\{0,1,2,\cdots,15\}$\\
    $\SC_{11}(16,1)$&(255,221)&35&33&3&$\{1,2,4,8,16,32,64,128\}\cup\{1,2,3,\cdots,32\}$\\
    $\SC_{01}(17,1)$&(255,219)&37&35&4&$\{1,2,4,8,16,32,64,128\}\cup\{0,1,2,\cdots,33\}$\\
    \hline
    $\SC_{11}(8,2)$&(255,231)&25&17&5&$\{1,2,\cdots,64,128,3,6,12,\cdots,192,129\}\cup\{1,2,\cdots,16\}$\\
    $\SC_{11}(16,2)$&(255,217)&39&33&5&$\{1,2,\cdots,64,128,3,6,12,\cdots,192,129\}\cup\{1,2,\cdots,32\}$\\
    \hline
  \end{tabular}
  \caption{Parameters of SRS codes with $n=255$ over GF(256)}
  \label{tab:SRS}
\end{table*}

The first four codes in Table \ref{tab:SRS} (with $t'=1$) meet the generalized Hamming weight bound and are MDS under the trace constraint. In general, for $n=255$, $m=8$, $d'=3$ and $d=2t+1$, the zero set works out to be $Z=\{1,2,4,8,16,32,64,128\}\cup\{1,2,3,\cdots,2t\}$. Hence, $|Z|=8+2t-(\lfloor\log_2(2t)\rfloor+1$ and $d=n-k+\lfloor\log_2(d-1)\rfloor-6$. Similarly, for $d'=4$, we get $d=n-k+\lfloor\log_2(d-2)\rfloor-6$. Therefore, the SRS codes have minimum distances close to the Singleton bound, particularly as $d$ increases. For both $d'=3$ and $d'=4$, these codes can be easily shown to meet the GHW bound. 

When the additional trace structure of SRS codes is used in the decoding, SRS codes turn out to be good competitors to RS codes offering good trade-offs between coding gain and complexity.
\section{List Decoders and Error-Correcting Properties}
\label{sec:decoders-srs-codes}
Since the minimum distance of the SRS code $\SC(t,t')$ of length $n=2^m-1$ symbols over GF($2^m$) is $2t+1$ in most cases, algebraic bounded distance decoding does not appear to be promising. Also, algebraically the trace operator is difficult to handle in a Berlekemp-Massey-like decoder based on simplifying power sums by Newton's identities. However, utilizing the structure of the image in a list decoder is beneficial as described below. Using the intuition gained from list decoders, we propose several soft decoders in later sections. 

Though an SRS code has a lesser minimum distance than an equal-rate RS code in many cases of interest, simple list decoders can be designed to correct a significant fraction of errors above half the minimum distance. In this section, we introduce and study list decoders for SRS codes, primarily as a means for studying the error-correcting capability of SRS codes. 
\subsection{List decoders}
Consider the SRS code $\SC(t,t')$ over $\twom$. As seen before, every codeword of the binary image of $\SC(t,t')$ can be written down as a $n\times m$ matrix with each column belonging to the $t'$-error-correcting binary BCH code. 

The proposed list decoder works as follows. The input to the decoder is the $n\times m$ matrix $R$ of received bits. Let $\underline{R}_i$ denote the $i$-th column of $R$. The first block of the decoder is a bounded-distance decoder for the $t'$-error correcting binary BCH code of length $n$. The BCH decoder runs on each column $\underline{R}_i$, $1\leq i\leq m$. The output of the $i$-th BCH decoder is denoted $\hat{\underline{R}}_i$. In case of decoder failure, $\hat{\underline{R}}_i=\underline{R}_i$. Let $\hat{R}$ denote the $n\times m$ matrix whose $i$-th column is $\hat{\underline{R}}_i$. The next step in the decoding is performed by a bank of $L$ $t$-error-correcting bounded-distance RS decoders. The $i$-th decoder ($1\leq i\leq L$) is parametrized by a set $S_i$, which is a subset of $\{1,2,\cdots,m\}$. The input to the $i$-th RS decoder is a $n\times m$ matrix whose $j$-th column is $\hat{\underline{R}}_j$ if $j\in S_i$ or $\underline{R}_j$ if $j\notin S_i$ ($1\leq j\leq m$). The matrix is converted to a $n\times1$ vector over $\text{GF}(2^m)$ for decoding by the $i$-th RS decoder.

Note that the set $S_i$ specifies the columns that are decoded by the $t'$-error-correcting binary BCH decoder before input to the $i^{th}$ RS decoder. Different RS decoders in the second step are parametrized by different $S_i$. The output from the $L$ RS decoders forms the list of possible codewords. The maximum list size is seen to be $2^m$.
\subsection{Analysis of the list decoder}
We devise a counting algorithm to calculate the fraction of weight-$w$ errors correctable by $\SC(t,t')$ using the proposed list decoder with list size set as $2^m$. For $w\leq t$, the fraction is 1. The calculation is done for $w>t$.

Let $P_m(w)$ denote the set of partitions of $w$ into not more than $m$ parts. Let $p$ be the partition given by $w=w_1+w_2+\cdots+w_l$ where $w_1\geq w_2\geq \cdots\geq w_l$. The numbers $w_1, w_2,\ldots, w_l$ denote the number of bit errors affecting $l$ out of the $m$ columns of the $n\times m$ codeword matrix. Equivalently, we can think of $w_1,w_2,\ldots,w_l$ as the weights of $l$ out of the $m$ columns of the $n\times m$ binary error matrix $E$.

For a given partition $p\equiv w_1+w_2+\cdots+w_l$ of $w$, an ensemble of error patterns $\mathcal{E}(p)$ exists with the column weight distribution $\{w_1,w_2,\ldots,w_l\}$. The size of the set $\mathcal{E}(p)$ is seen to be
$$\vert\mathcal{E}(p)\vert=\frac{l!}{n_1!n_2!\cdots n_r!} \binom{m}{l}\prod^{l}_{i=1}\binom{n}{w_i},$$
where $r$ is the number of distinct weights in the set of weights $\{w_1,w_2,\ldots,w_l\}$, and $n_i$ is the number of times the $i$-th distinct weight occurs in the set of weights. For instance, if the set of weights is $\{4,3,3,1,1\}$, then $r=3$, $n_1=1$, $n_2=2$, and $n_3=2$.

Thus, the fraction of correctable errors for weight $w$, denoted $f_w$ is given by
$$f_w=\frac{\sum_{p}P_c(p)\vert\mathcal{E}(p)\vert}{\binom{nm}{w}},$$
where $P_c(p)$ is the probability that an error vector with column weight distribution $p$ is correctable. 

To determine $P_c(p)$, the partitions in $P_m(w)$ are modified by deleting the parts that are lesser than $t'$ to account for the BCH decoder. Since the list size is $2^m$, there exists an RS decoder parametrized by the set of columns corresponding to the parts in $p$ of weight less than $t'$. For example, let $t'=1$ and $w=9$. Let $p$ be the partition given by $9=4+3+1+1$; $p$ is modified as $\hat{p}$ given by $\hat{p}\equiv4+3$. Hence, a suitable RS decoder will see an error matrix with column weight distribution $\hat{p}$. Each partition in $P_m(w)$ is modified in a similar way to form a set $\hat{P}_m(w)$. Let $\hat{p}$ be given by $\hat{p}\equiv w_1+w_2+\cdots+w_k$. The sum $\hat{w}=w_1+w_2+\cdots+w_k$ need not be equal to $w$; it is less than or equal to $w$. Based on the modified partition $\hat{p}$, we have four different cases. 
\begin{enumerate}
\item If $\hat{p}$ is empty, it implies that all elements in the partition $p$ were $\leq t'$. A suitable RS decoder will output the correct codeword, and $P_c(p)=1$.
\item If $\hat{w}\leq t$, then whatever way errors are distributed along different columns, the total number of rows affected cannot exceed $t$. A suitable RS decoder will output the correct codeword, and $P_c(p)=1$. 
\item If $w_1>t\geq t'$, then more than $t$ rows will be in error for all RS decoders. By the bounded-distance property, we assume that such error patterns can never be corrected, and $P_c(p)=0$.
\item If $\hat{p}$ does not fall into any of the above three categories, the error pattern may or may not be correctable depending on how the errors are distributed along the columns. For this case, a more detailed analysis is necessary. In this case, $0<P_c(p)<1$.
\end{enumerate}

For Case 4 above, the computation of $P_c(p)$ is done as follows. An error matrix $E\in \mathcal{E}(p)$ for $\hat{p}\equiv w_1+w_2+\cdots+w_k$ is modeled by a discrete random process that involves $k$ steps. The $i$-th step corresponds to the random placement of $w_i$ ones in one of the $m$ columns. Let $\{Y_1,Y_2,\ldots,Y_k\}$ be a sequence of discrete random variables, where $Y_i$ denotes the total number of nonzero rows of $E$ after the $i$-th step. For instance, $Y_1$ denotes the number of nonzero rows of $E$ after the first step, which will be $w_1$ with probability $1$. $Y_2$ denotes the number of nonzero rows after the second step. $Y_2$ takes values from $w_1$ to $(w_1+w_2)$ with different probabilities. The probability mass function (PMF) of $Y_2$ can be determined from the PMF of $Y_1$ and the value $w_2$. Similarly, we can find the PMFs of all the random variables $Y_1$ to $Y_k$ starting from the PMF of $Y_1$ and the values $w_1,w_2,\ldots,w_k$. Finally,
$$P_c(p)=\text{Prob}\{Y_k\leq t\}.$$

Fig. \ref{fig:analysis} shows a comparison of the 8-error-correcting (255, 239, 17) RS code ($\SC(8)$) over GF(256) and the (255, 239, 13) SRS code ($\SC(6,1)$) over GF(256). 
\begin{figure}[htb]
  \centering
  \subfigure[D$L$: list decoder of size $L$]{
  \includegraphics[width=0.4\textwidth]{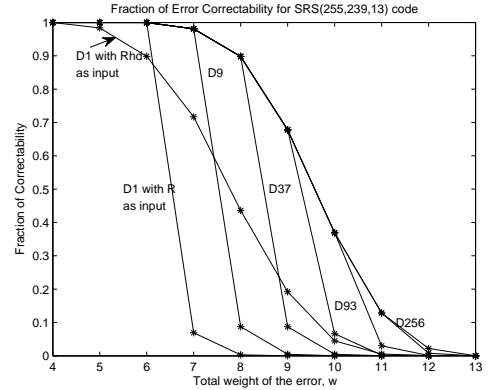}}
  \subfigure[Block-error rate plot]{
  \includegraphics[width=0.45\textwidth]{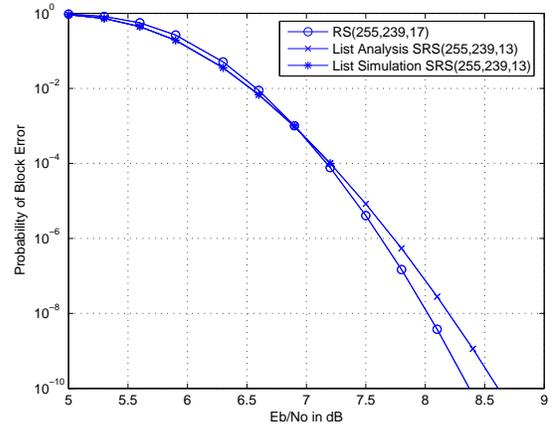}}
  \caption{Comparison of $\SC(6,1)$ and $\SC(8)$ over GF(256).}
  \label{fig:analysis}
\end{figure}
The list decoder was simulated over an AWGN channel with hard-decision decoding to verify the analysis. We see that the analysis matches with the simulated list decoder in the block-error rate plot, and the SRS code is competitive with the RS code of same rate down to a block-error rate of $10^{-10}$.

Notice that the list decoder D256 (see Fig. \ref{fig:analysis}(a)) corrects a significant fraction of weight-7, 8 and 9 errors though the minimum distance of the code is 13. It is interesting to note that D1 fails to correct some weight-6 errors because of errors in the Hamming decoders in the first step of decoding. An important factor in successful decoding is the choice of suitable columns of the received vector that need to be decoded by a Hamming decoder. We propose to use soft information from the channel for making suitable decisions in the first stage and develop practical decoders for SRS codes. 
\section{Soft-input Decoders}
\label{sec:soft-input-decoders}
Because of the special structure of SRS codes, several suboptimal soft decoders of varying complexity are possible. We propose three types of soft-input decoders of increasing complexity. The codes $\SC_{11}(6,1)$ and $\SC_{01}(6,1)$ are compared with $\SC(8)$ over GF(256) in our simulations. Soft decoders for other codes yield similar gains.

We assume BPSK modulation ($0\rightarrow+1,1\rightarrow-1$) over an AWGN channel with variance $\sigma^2$. The standard $Q$ function, defined as $Q(x)=\frac{1}{\sqrt{2\pi}}\int_{x}^{\infty}e^{\frac{-x^2}{2}}dx$, is used in describing the decoder. 

For a SRS code $\SC(t,t')$ of length $n=2^m-1$ over GF$(2^m)$, the received information $R$ is a $n\times m$ real-valued matrix and let $R_{i,j}$ denotes the value in the $i$-th row and $j$-th column of $R$. The proposed soft-input decoders work in two stages. The first stage decodes the columns of $R$ according to the trace code. We restrict ourselves to $d'=3$ (Hamming code) and $d'=4$ (even-weight subcode of Hamming code) for simplicity. The second stage decodes the output of the first stage according to the $t$-error-correcting RS code over GF$(2^m)$.
\subsection{Soft-guided decoders}
We begin with a low-complexity soft-input decoder, which we call a soft-guided decoder. In the first stage of a soft-guided decoder for SRS codes, hard-decision syndromes for the trace code (Hamming or its even-weight subcode) are computed for each of the $m$ columns of $R$. Depending on the trace code, the following possibilities occur:
\begin{enumerate}
\item $d'=3$: If the syndrome for the $i$-th column is non-zero and indicates an error in location $e$ and $|R_{e,i}|<\Delta$, the location is confirmed to be in error; otherwise, the location is assumed to be error-free. The threshold, denoted $\Delta$, is heuristically chosen to satisfy
\begin{multline*}
\binom{n}{1}p(1-p)^{n-1}\frac{Q\left(\frac{1+\Delta}{\sigma}\right)}{Q\left(\frac{1}{\sigma}\right)}=\\
\binom{n}{2}p^2(1-p)^{n-2}\frac{Q\left(\frac{1-\Delta}{\sigma}-Q\left(\frac{1}{\sigma}\right)\right)}{1-Q\left(\frac{1}{\sigma}\right)},
\end{multline*}
which equates the (approximate) probabilities of single errors resulting in no confirmation to double errors resulting in erroneous confirmation.   
\item$d'=4$: In this case, we can detect double errors. If the syndrome for the $i$-th column is non-zero and indicates a double error, no error locations are confirmed. If the syndrome indicates an error in location $e$ and $|R_{e,i}|<\Delta$, the location is confirmed to be in error. The threshold $\Delta$ is chosen to satisfy
\begin{multline*}
\binom{n}{1}p(1-p)^{n-1}\frac{Q\left(\frac{1+\Delta}{\sigma}\right)}{Q\left(\frac{1}{\sigma}\right)}=\\
\binom{n}{3}p^3(1-p)^{n-3}\frac{Q\left(\frac{1-\Delta}{\sigma}-Q\left(\frac{1}{\sigma}\right)\right)}{1-Q\left(\frac{1}{\sigma}\right)},
\end{multline*}
which equates the (approximate) probabilities of single errors resulting in no confirmation to triple errors resulting in erroneous confirmation.   
\end{enumerate}
Hard decisions are made on $R$, and the confirmed error locations are flipped. The output is a single $n\times m$ binary matrix. Note that several other similar suboptimal first stages with varying complexity can be designed.

The second stage involves one $t$-error-correcting bounded-distance RS decoder working on the output of the first stage. The performance of the soft-guided decoder is shown in Fig. \ref{fig:ber2}. We see that the performance of a simple soft-guided decoder for the SRS code is comparable to that of the hard-decision decoder for the MDS RS code at the same rate. Notice that the code $\SC_{01}(6,1)$ performs marginally better than $\SC_{11}(6,1)$ because of the identification of double errors.
\begin{figure}[!h]
	\centering
	\includegraphics[angle=0,width=0.45\textwidth]{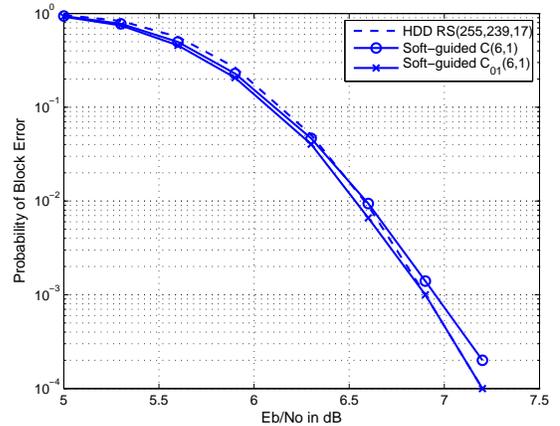}
	\caption{Performance of soft-guided decoder.}
	\label{fig:ber2}
\end{figure}
\subsection{Hybrid decoders}
In hybrid soft-input decoders, the first stage is an optimal soft decoder for the trace code. An efficient implementation for bitwise-MAP decoders for Hamming codes and their even-weight subcodes can be found in \cite{Ashikhmin:2004wv}\cite{Thangaraj:2006qr}. The complexity of these decoders is $O(n\log n)$, where $n$ is the blocklength. These decoders are implementable in hardware through transformations such as the Walsh-Hadamard transform. We skip the details of the implementation. 

In the first stage, an efficient MAP decoder is run on each column of $R$ to obtain log-likelihood ratios (LLRs) for each bit conditioned on the received values in the corresponding column (for a bit in the $i$-th column, the received values in $\underline{R}_i$ are used). 

After the first stage, hard decisions are made on the LLRs to obtain a single $n\times m$ binary matrix. The second stage is a $t$-error-correcting bounded-distance RS decoder. We readily see that the complexity of the first stage in hybrid decoders is higher than that of soft-guided decoders.

The performance of hybrid decoders is shown in Fig. \ref{fig:ber3}. We see that the hybrid decoders provide a coding gain of more than 0.6 dB over hard-decision decoders of MDS RS codes at the same rate. We also notice that additional gain is obtained by using $\SC_{01}(6,1)$ with $d'=4$. The gain is about 0.7 dB at a block error rate of $10^{-3}$.
\begin{figure}[!h]
	\centering
	\includegraphics[angle=0,width=0.45\textwidth]{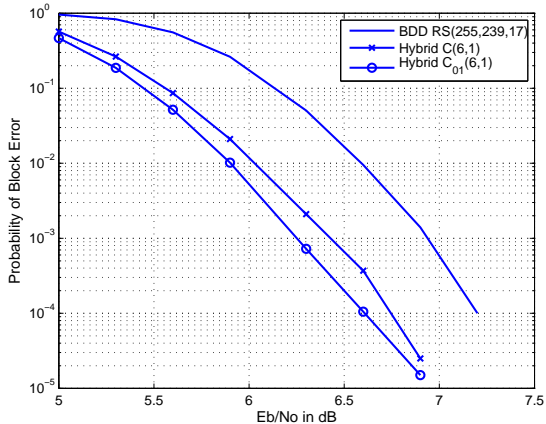}
	\caption{Performance of hybrid decoder.}
	\label{fig:ber3}
\end{figure}

When compared to the bit-level GMD algorithm \cite{Jiang:2008cr}, the hybrid decoder appears to be simpler in complexity. The soft processing in bit-level GMD involves sorting, which is comparable to the complexity of evaluating the Walsh-Hadamard transform. However, the hybrid decoder uses the traditional Berlekemp-Massey bounded-distance decoder only once, while the bit-level GMD employs the Koetter-Vardy (KV) soft-input decoder for RS codes iteratively.

A weakness of the hybrid decoder is that bounds for very low block error rates are difficult to prove, unlike the bit-level GMD. The error-correcting capability of SRS codes under hard-decision list decoding, as depicted in Fig. \ref{fig:analysis}, seems to suggest that the performance of hybrid decoders should extend to lower block error rates as well. 
\subsection{Soft decoders}
We call the most complex among the proposed soft-input decoders as simply soft decoders. In the first stage, we employ efficient implementations of the optimal bitwise MAP-decoders for the trace (similar to hybrid decoders). In the second stage, the Koetter-Vardy (KV) soft-input decoder for RS codes presented in \cite{Koetter:2003ph}\cite{W.-J.-Gross:2002mj} is employed. The LLRs obtained after the first stage are converted to suitable inputs to the KV decoder using the methods suggested in \cite{W.-J.-Gross:2002mj}. We skip the details of the implementation, since we closely follow the ideas in \cite{W.-J.-Gross:2002mj} in our simulations. 

The performance of soft decoders is depicted in Fig. \ref{fig:ber4}. 
\begin{figure}[!h]
	\centering
	\includegraphics[angle=0,width=0.45\textwidth]{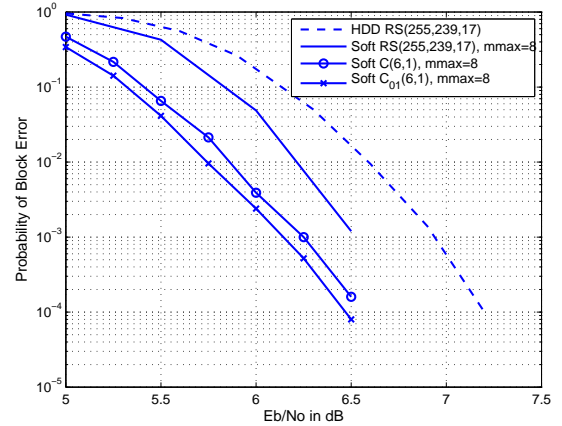}
	\caption{Performance of soft decoder.}
	\label{fig:ber4}
\end{figure}
We see that gains of about 0.8-0.9 dB over comparable hard-decoded RS codes are possible with soft decoders. Gains of about 0.4-0.5 dB are obtained over KV soft decoding of RS codes of same rate. The parameter `mmax' (from \cite{Koetter:2003ph}) indicates the complexity of the second stage. 

The complexity of the proposed soft decoder is roughly comparable to that of the bit-level GMD decoder, provided the iterations of the KV soft decoder (in bit-level GMD) are carefully optimized. The gain of the proposed soft decoder is marginally better than bit-level GMD.

In summary, for the code $\SC_{01}(6,1)$, we observe that soft-guided decoders appear to be similar in performance to MDS RS codes at the same rate. We see gains of about 0.7 dB over hard-decision RS decoders with limited complexity hybrid decoders. For more complex soft decoders, we observe gains of 0.4-0.5 dB over soft KV decoders. 
\section{Conclusion}
\label{sec:conclusion}
In this work, we proposed and studied a new approach for obtaining higher coding gains in situations where traditional Reed-Solomon codes have been used so far - namely, rate about 0.9 over GF(256). The approach suggests the use of a suitably chosen subcode of Reed-Solomon codes. This subcode is characterized by the property that its trace code has a minimum distance larger than 1. Using the properties of the trace and image, we showed that additional coding gain can be obtained by efficiently processing soft values. Gains of about 0.7-0.8 dB are possible over bounded-distance decoders of traditional RS codes with low complexity soft decoders such as the proposed hybrid decoder. When compared to other soft decoders for RS codes in the literature, a gain of 0.4-0.5 dB is possible with the proposed soft decoder for SRS codes.  

This work demonstrates the practical utility of studying the properties of trace and image of codes over non-binary fields. Several avenues are possible for extending this study both from a theoretical and practical viewpoint.
\bibliographystyle{IEEEtran}
\bibliography{docdblatex}

\begin{thebibliography}{10}
\providecommand{\url}[1]{#1}
\csname url@samestyle\endcsname
\providecommand{\newblock}{\relax}
\providecommand{\bibinfo}[2]{#2}
\providecommand{\BIBentrySTDinterwordspacing}{\spaceskip=0pt\relax}
\providecommand{\BIBentryALTinterwordstretchfactor}{4}
\providecommand{\BIBentryALTinterwordspacing}{\spaceskip=\fontdimen2\font plus
\BIBentryALTinterwordstretchfactor\fontdimen3\font minus
  \fontdimen4\font\relax}
\providecommand{\BIBforeignlanguage}[2]{{%
\expandafter\ifx\csname l@#1\endcsname\relax
\typeout{** WARNING: IEEEtran.bst: No hyphenation pattern has been}%
\typeout{** loaded for the language `#1'. Using the pattern for}%
\typeout{** the default language instead.}%
\else
\language=\csname l@#1\endcsname
\fi
#2}}
\providecommand{\BIBdecl}{\relax}
\BIBdecl

\bibitem{Reed:1960qf}
I.~S. Reed and G.~Solomon, ``Polynomial codes over certain finite fields,''
  \emph{J. SIAM}, vol.~8, pp. 300--304, 1960.

\bibitem{Guruswami:1999df}
V.~Guruswami and M.~Sudan, ``Improved decoding of {R}eed-{S}olomon and
  algebraic-geometry codes,'' \emph{Information Theory, IEEE Transactions on},
  vol.~45, no.~6, pp. 1757--1767, Sep 1999.

\bibitem{Koetter:2003ph}
R.~Koetter and A.~Vardy, ``Algebraic soft-decision decoding of {R}eed-{S}olomon
  codes,'' \emph{Information Theory, IEEE Transactions on}, vol.~49, no.~11,
  pp. 2809--2825, 2003.

\bibitem{Jiang:2008cr}
J.~Jiang and K.~R. Narayanan, ``Algebraic soft-decision decoding of
  {R}eed--{S}olomon codes using bit-level soft information,'' \emph{Information
  Theory, IEEE Transactions on}, vol.~54, no.~9, pp. 3907--3928, Sept. 2008.

\bibitem{Forney:1966qv}
D.~Forney, ``Generalized minimum distance decoding,'' \emph{Information Theory,
  IEEE Transactions on}, vol.~12, no.~2, pp. 125--131, 1966.

\bibitem{Vardy:1994fq}
A.~Vardy and Y.~Be'ery, ``Maximum-likelihood soft decision decoding of {BCH}
  codes,'' \emph{Information Theory, IEEE Transactions on}, vol.~40, no.~2, pp.
  546--554, Mar 1994.

\bibitem{Ponnampalam:2002jy}
V.~Ponnampalam and B.~Vucetic, ``Soft decision decoding of {R}eed-{S}olomon
  codes,'' \emph{Communications, IEEE Transactions on}, vol.~50, no.~11, pp.
  1758--1768, 2002.

\bibitem{Jiang:2004pz}
J.~Jiang and K.~R. Narayanan, ``Iterative soft decoding of {R}eed-{S}olomon
  codes,'' \emph{IEEE Communications Letters}, vol.~8, no.~4, pp. 244--246,
  2004.

\bibitem{MacWilliams:1977rw}
F.~J. MacWilliams and N.~J.~A. Sloane, \emph{The theory of error-correcting
  codes}.\hskip 1em plus 0.5em minus 0.4em\relax Amsterdam, The Netherlands:
  North-Holland, 1977.

\bibitem{Wei:1991if}
V.~Wei, ``Generalized hamming weights for linear codes,'' \emph{Information
  Theory, IEEE Transactions on}, vol.~37, no.~5, pp. 1412--1418, Sep 1991.

\bibitem{Helleseth:1995si}
T.~Helleseth, T.~Klove, V.~Levenshtein, and O.~Ytrehus, ``Bounds on the minimum
  support weights,'' \emph{Information Theory, IEEE Transactions on}, vol.~41,
  no.~2, pp. 432--440, Mar 1995.

\bibitem{Augot:1992cl}
D.~Augot, P.~Charpin, and N.~Sendrier, ``Studying the locator polynomials of
  minimum weight codewords of {BCH} codes,'' \emph{Information Theory, IEEE
  Transactions on}, vol.~38, no.~3, pp. 960--973, 1992.

\bibitem{Ashikhmin:2004wv}
A.~Ashikhmin and S.~Litsyn, ``Simple {MAP} decoding of first-order
  {R}eed-{M}uller and {H}amming codes,'' \emph{Information Theory, IEEE
  Transactions on}, vol.~50, no.~8, pp. 1812--1818, August 2004.

\bibitem{Thangaraj:2006qr}
A.~Thangaraj, ``Simple map decoding of binary cyclic codes,'' \emph{Information
  Theory, 2006 IEEE International Symposium on}, pp. 464--468, July 2006.

\bibitem{W.-J.-Gross:2002mj}
W.~J. Gross, F.~R. Kschischang, R.~Koetter, and P.~G. Gulak, ``Simulation
  results for algebraic soft-decision decoding of {R}eed-{S}olomon codes,'' in
  \emph{Proceedings of the 21st Biennial Symposium on Communications}, Queen's
  University, Kingston, Ontario, June 2-5 2002, pp. 356--360.

\end{thebibliography}
\end{document}